\pdfoutput=1
\documentclass[10pt]{article}

\usepackage[letterpaper,textwidth=5.5in,textheight=9in,centering]{geometry}
\usepackage[T1]{fontenc}
\usepackage{times}
\usepackage{microtype}
\usepackage[authoryear,round]{natbib}
\setcitestyle{citesep={;},aysep={,},yysep={;}}
\usepackage{titlesec}
\titleformat*{\section}{\large\bfseries}
\titleformat*{\subsection}{\normalsize\bfseries}
\titleformat*{\subsubsection}{\normalsize\bfseries\itshape}
\titlespacing*{\section}{0pt}{2.2ex plus .5ex minus .2ex}{1.2ex plus .3ex minus .2ex}
\titlespacing*{\subsection}{0pt}{1.8ex plus .5ex minus .2ex}{0.8ex plus .2ex}
\titlespacing*{\subsubsection}{0pt}{1.5ex plus .5ex minus .2ex}{0.5ex plus .2ex}
\renewenvironment{abstract}{%
  \small
  \begin{center}{\bfseries\abstractname\vspace{-.5em}\vspace{0pt}}\end{center}%
  \quote}{\endquote}

\usepackage{hyperref}
\usepackage{url}
\usepackage{amsmath}
\usepackage{amssymb}
\usepackage{amsthm}
\newtheorem{proposition}{Proposition}
\usepackage{arydshln}
\usepackage{booktabs}
\usepackage{multirow}
\usepackage{xcolor}
\usepackage{graphicx}
\usepackage{algorithm}
\usepackage{algpseudocode}
\usepackage[capitalise]{cleveref}
\newcommand{\gain}[1]{\textcolor{teal}{\ensuremath{\uparrow #1}}}
\newcommand{\loss}[1]{\textcolor{purple}{$\downarrow$#1}}

\newcommand{\method}{RidgeRank}

\definecolor{linkred}{rgb}{0.60,0.10,0.10}
\definecolor{citeblue}{rgb}{0.10,0.25,0.60}
\hypersetup{
  colorlinks=true,
  linkcolor=linkred,
  citecolor=citeblue,
  urlcolor=citeblue,
  pdftitle={\method: Fast Visual Document Reranking via Score Fusion and a Shallow Linear Readout},
  pdfauthor={Shubing Yang, Dongfang Zhao}
}

\title{\method: Efficient Visual Document Reranking via Score Fusion and a Shallow Linear Readout}

\author{%
  Shubing Yang\\
  University of Washington\\
  \texttt{sueyoung@uw.edu}
  \and
  Dongfang Zhao\\
  University of Washington\\
  \texttt{dzhao@cs.washington.edu}%
}

\date{}

\begin{document}

\maketitle

\begin{abstract}
Multimodal language models rerank visual document retrieval results accurately, but scoring every candidate page at full cost makes them slow. Some methods that compress these rerankers need relevance labels to regain accuracy, and they rank by the reranker score alone. \method{} measures how much relevance signal the reranker score lacks and recovers it from the retriever score through a closed-form fusion rule. Maximizing a correlation objective gives the optimal fusion weight, along with the exact condition under which the reranker score by itself cannot reach that optimum. The reranker is further corrected by a single vector applied to an intermediate hidden state, obtained through one centered ridge regression onto the same model's full-depth scores on uncompressed pages. On 12 datasets drawn from ViDoRe 2 and ViDoRe 3, evaluated with two retrievers and two language model backbones, \method{} brings NDCG@5 to within 1.2\,pp of a full cross encoder with speedups of up to 48 times, advancing the accuracy and latency Pareto frontier for visual document reranking.
\end{abstract}

\section{Introduction}

Multimodal language model rerankers improve ranking quality in visual document retrieval, but their cost grows with the number of pages scored. Visual document retrieval supports question answering and generation over long, visually rich documents \citep{visrag,m3docrag,vdocrag,mmdocir} in domains such as biomedical lectures, economics reports and ESG disclosures \citep{vidore2}. These applications require low response latency, yet an uncompressed reranker processes each candidate page as a full image.

Several methods that accelerate these rerankers rely on relevance labels to regain the accuracy lost to compression, and they rank with the reranker score alone. They reduce inference cost by caching page prefixes offline \citep{prettr,deformer,minireranker}, stopping at intermediate layers \citep{deebert,calm,see,e2rank}, and subsampling or quantizing stored visual positions \citep{rankprune,kivi,ziprerank}. These limitations motivate calibrating the compressed reranker with the full model's scores and incorporating retriever scores into the final ranking.

At the core of \method{} is a closed-form fusion rule that quantifies the relevance information absent from the reranker score and supplies it from the retriever score. The two scores are normalized to zero mean and unit variance within each candidate list and ranked by their weighted combination. The operating point is the weight that maximizes the correlation of this combination with relevance, and it depends only on the correlation of each score with relevance and the correlation between the two scores. The same derivation yields the exact condition under which the retriever carries information about relevance beyond the reranker and should receive a positive weight.

A second component corrects the reranker through one vector applied to a shallow hidden state, which a single centered ridge solve fits to reproduce the score the model gives without compression. That score is the margin between the two answer token logits and equals an inner product with the model's final hidden state, which makes a linear readout of an earlier state the natural object to fit. Centering the fit within each candidate list lets it learn how candidates of the same query differ from one another, the only information the ranking uses. Its ridge strength is chosen by the same centered error on held out queries, and the whole fit is therefore determined by the model's own scores while the backbone stays frozen.

Evaluation spans 12 datasets with two retrievers and two multimodal backbones, where \method{} gets within 1.2\,pp of the full cross encoder in NDCG@5 at up to 48 times the speed.

The paper makes three contributions.

\begin{itemize}
    \item A correlation objective quantifies the relevance information that the reranker score lacks. Its optimal operating point has a closed form, and the same analysis gives the exact condition under which the reranker score by itself cannot attain it. (\Cref{sec:fusion})
    \item \method{} calibrates the reranker with a shallow linear readout, fitted to the model's own scores by one centered ridge solve with no gradient updates and no relevance judgements. A bound links agreement with this target to correlation with relevance. (\Cref{sec:readout,sec:objective})
    \item On 12 ViDoRe 2 and ViDoRe 3 datasets, with two retrievers and two backbones (Qwen2.5-VL-7B and GLM-4.1V-9B), \method{} comes to within 1.2\,pp NDCG@5 of a full cross encoder at up to 48 times lower latency. (\Cref{sec:eval})
\end{itemize}

\section{Related Work}

\textbf{Multimodal document reranking.} Late interaction over image patches removed the parsing stage from document retrieval and set the standard for first stage recall \citep{colbert,colpali}, and the ViDoRe benchmarks made page level retrieval the common ground for comparing such systems \citep{colpali,vidore2}. A second stage built on a multimodal language model improves the order further. The prevailing form is pointwise. Each candidate is scored on its own and the resulting score replaces the retriever's \citep{monoqwen,minireranker}, and the strongest current reranker orders a list by the margin between its yes and no tokens alone \citep{qwen3vlrerank}. Hybrid systems combine two retrievers in representation space \citep{gqr}, and test time methods feed the reranker score back into the query representation and retrieve again \citep{refit,tour}. We instead characterize, for a fixed candidate set, the relevance information that the reranker discards and the retriever score still carries.

\textbf{Efficient inference for rerankers.} Query independent parts of the input are encoded once and cached, so a request runs only the query against a stored prefix \citep{prettr,deformer}. The forward pass is halted early, either adaptively per instance or at a layer fixed in advance \citep{deebert,calm,see,e2rank}. What is stored is reduced by discarding visual positions and by quantizing the tensors that remain \citep{rankprune,kivi}. A listwise variant drops autoregressive decoding by scoring a whole candidate list in one pass \citep{ziprerank}, and recent multimodal rerankers combine caching, early exit and cache compression to report large savings on ViDoRe \citep{minireranker}. We adopt these techniques unchanged and claim none of them. Some methods train the compressed model to recover accuracy. We fit a readout through a closed form solve on unlabelled queries.

\textbf{First stage scores and layerwise readouts.} Neural and first stage scores are combined by rank \citep{rrf} or by a weighted sum whose weight is tuned on labelled queries \citep{fusionfunctions,wanglin,fastforward}, or the first stage score is fed into the reranker as an input \citep{askari}. In these works the weight is a hyperparameter searched anew for each collection, and when the first stage score adds information beyond the reranker is left unexplained. We instead derive the optimal operating point in closed form, together with the exact condition under which the reranker score alone is insufficient. Linear readouts of intermediate states have been fitted to reproduce the final layer for interpretation \citep{tunedlens}, trained into the network by supervising shallow layers with deep ones \citep{byot,fastbert,nanovdr}, or trained contrastively on the layers of a frozen retriever \citep{miner}. Each of them is obtained by gradient training and corrects a single change to the network. Our readout is a single vector solved in closed form against the frozen reranker's own uncompressed score, and it compensates for the reduced depth and the compressed input at once.

\begin{figure}[t]
\centering
\includegraphics[width=\textwidth]{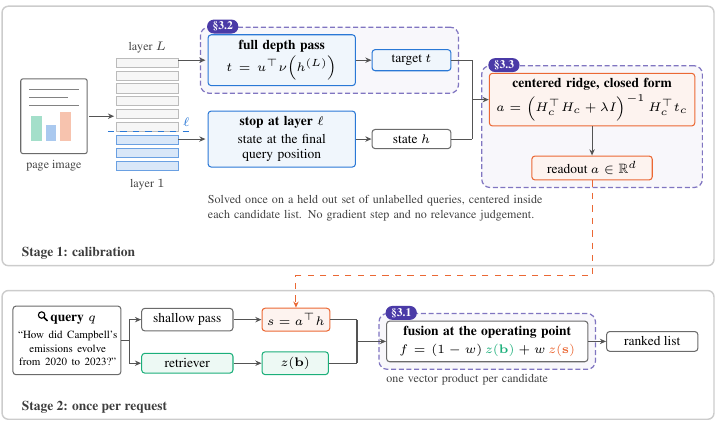}
\caption{\method. \textbf{Stage 1} runs once on unlabelled query page pairs. A full depth pass on the uncompressed page gives the target $t$, a pass stopped at layer $\ell$ gives the state $h$, and one centered ridge solve turns the pairs into the readout $a$. \textbf{Stage 2} runs once per request. A pass stopped at layer $\ell$ gives the state of each candidate, $a$ scores it with one vector product, and that score is fused with the retriever's at the operating point.}
\label{fig:overview}
\end{figure}

\section{Method}
\label{sec:method}

A two stage retrieval system returns a set $\mathcal{C}_q$ of $K$ candidate pages and retriever scores $\mathbf{b}_q \in \mathbb{R}^{K}$ for query $q$. A multimodal language model then scores these candidates to produce $\mathbf{s}_q \in \mathbb{R}^{K}$. We combine both scores to determine the final ranking. The retriever fixes the candidate set, so reranking changes only its order and requires no additional retrieval.

Following prior work \citep{prettr,deformer,minireranker}, we encode the query independent part of the prompt, comprising the page image and a fixed instruction, once offline and store the first $\ell$ layers in int8. We keep a fraction $\gamma$ of the $n_{\mathrm{vis}}$ visual positions \citep{rankprune,kivi}. At inference, only the query suffix runs over the stored prefix. We treat $\ell$ and $\gamma$ as fixed serving settings rather than parameters selected by our method.

\Cref{fig:overview} summarizes the method. \Cref{sec:fusion} derives the fusion weight, \Cref{sec:readout} defines the self-supervised target, and \Cref{sec:objective} gives the centered ridge solution.

\subsection{Fusion as Operating Point}
\label{sec:fusion}

We rank by a convex combination of the two scores, each standardized inside the candidate set because the two are on different scales that vary from query to query. The standardization is

\begin{equation}
  z(\mathbf{v})_i \;=\; \frac{v_i - \mu(\mathbf{v})}{\sigma(\mathbf{v})},
  \qquad \mathbf{v} \in \mathbb{R}^{K},
  \label{eq:zscore}
\end{equation}

where $\mu$ and $\sigma$ are the mean and standard deviation over the $K$ entries, and the ranking is by

\begin{equation}
  f_{q,p} \;=\; (1-w)\, z(\mathbf{b}_q)_p \;+\; w\, z(\mathbf{s}_q)_p,
  \qquad w \in [0,1].
  \label{eq:fusion}
\end{equation}

At $w=1$ the retriever score is discarded and ranking uses the reranker score alone, as in pointwise multimodal reranking. At $w=0$ the reranker is ignored. The following result characterizes when the correlation objective favors an interior weight.

\begin{proposition}
\label{prop:weight}
Write $\mathbf{y} \in \mathbb{R}^{K}$ for the standardized graded relevance of a candidate set and set
\begin{equation}
  c_b = \tfrac{1}{K}\langle z(\mathbf{b}), \mathbf{y}\rangle, \qquad
  c_s = \tfrac{1}{K}\langle z(\mathbf{s}), \mathbf{y}\rangle, \qquad
  \rho = \tfrac{1}{K}\langle z(\mathbf{b}), z(\mathbf{s})\rangle ,
  \label{eq:corrs}
\end{equation}
so that $c_b$ and $c_s$ are the correlations of the two scores with relevance and $\rho$ is their correlation with each other. If $|\rho| < 1$, $c_b > \rho c_s$, and $c_s > \rho c_b$, the weight maximizing the correlation between $\mathbf{f}$ of Equation \ref{eq:fusion} and $\mathbf{y}$ over $w \in [0,1]$ is
\begin{equation}
  w^{\star} \;=\; \frac{c_s - \rho\, c_b}{(c_b + c_s)(1-\rho)},
  \label{eq:wstar}
\end{equation}
which lies strictly inside $(0,1)$.
\end{proposition}

\begin{proof}[Proof sketch]
Writing $\alpha = 1-w$ and $\beta = w$, the correlation between $\mathbf{f}$ and $\mathbf{y}$ is $(\alpha c_b + \beta c_s)/\sqrt{\alpha^{2} + \beta^{2} + 2\alpha\beta\rho}$, a linear form divided by the norm induced by the score correlation matrix $\Sigma$. This matrix has diagonal entries one and off diagonal entries $\rho$; $|\rho|<1$ ensures that it is positive definite. By the Cauchy Schwarz inequality the correlation is maximized along $\Sigma^{-1}(c_b,c_s)^{\top}$, proportional to $(c_b-\rho c_s,\;c_s-\rho c_b)$. Both components are positive under the stated conditions, and rescaling them to $\alpha+\beta=1$ gives Equation \ref{eq:wstar}. \Cref{sec:proof1} gives the full proof.
\end{proof}

The condition $c_b > \rho c_s$ says that the retriever contributes information about relevance beyond its correlation with the reranker score. The difference $c_b-\rho c_s$ is the numerator of the retriever's partial correlation with relevance after accounting for the reranker. The corresponding condition $c_s > \rho c_b$ ensures that the reranker also receives positive weight. When both hold, ranking by either score alone gives a lower value of the correlation objective. These conditions concern the contribution of each score to relevance, not whether the two scores assign similar values to individual candidates. We read the three quantities of Equation \ref{eq:corrs} as averages over a corpus, each computed with an inner product weighted by the discount of the ranking metric, $\omega_i=1/\log_2(1+\text{rank}_i)$ under the retriever's order. The statement remains valid with this weighted inner product (\Cref{sec:proof1}). We deploy $w = w^{\star}$, with every vector standardized and the three correlations read in this weighted inner product on another corpus of the same benchmark, so the corpus being ranked contributes no relevance judgement to its own weight. The fusion itself keeps the standardization of Equation \ref{eq:zscore}.

\subsection{Self-Supervised Target}
\label{sec:readout}

The reranker score is a linear readout of the state at layer $\ell$. We run the query suffix over the stored prefix to that layer and take the hidden state at the final suffix position, $h_{q,p} \in \mathbb{R}^{d}$, where $d$ is the backbone width, so that

\begin{equation}
  s_{q,p} \;=\; a^{\top} h_{q,p}, \qquad a \in \mathbb{R}^{d},
  \label{eq:readout}
\end{equation}

whose cost at request time is one vector product per candidate, negligible beside the forward pass that produces $h_{q,p}$.

The target is the score the same backbone gives at full depth on the uncompressed page. Let $L$ be the final layer, let $h^{(L)}_{q,p}$ be the state at the final position of that uncompressed forward, and write $o_{\texttt{yes}}$ and $o_{\texttt{no}}$ for the logits of the two answer tokens. The target is the decision margin

\begin{equation}
  t_{q,p} \;=\; o_{\texttt{yes}}(q,p) \;-\; o_{\texttt{no}}(q,p)
           \;=\; u^{\top} \nu\big(h^{(L)}_{q,p}\big),
  \qquad u = u_{\texttt{yes}} - u_{\texttt{no}},
  \label{eq:teacher}
\end{equation}

where $u_{\texttt{yes}}$ and $u_{\texttt{no}}$ are the unembedding rows of the two answer tokens and $\nu$ is the final normalization the backbone applies before unembedding. The full depth model is used only during offline calibration.

Equation \ref{eq:teacher} motivates a linear readout but does not imply an exact linear map from the deployed state to the target. The final score is a linear function of the normalized final state; both the intervening layers and the final normalization can be nonlinear. The deployed state $h_{q,p}$ differs from $h^{(L)}_{q,p}$ because depth is truncated from $L$ to $\ell$ and the stored prefix is subsampled and quantized. Writing $W$ for a linear surrogate of the map from the deployed state to the uncompressed one on the subspace spanned by calibration features, and $\nu'$ for a local linear approximation to $\nu$ on that subspace, we obtain

\begin{equation}
  t_{q,p} \;\approx\; u^{\top} \nu' W\, h_{q,p}
           \;=\; \big(W^{\top} \nu'^{\top} u\big)^{\top} h_{q,p},
  \label{eq:collapse}
\end{equation}

which has the form of Equation \ref{eq:readout} with $a=W^{\top}\nu'^{\top}u$. This approximation combines the effects of depth truncation and cache compression in $W$. The readout estimates only the direction of that map relevant to the target margin, rather than the entire final state. Thus one fitted vector can compensate for both changes when the linear approximation holds on the calibration distribution. Any constant offset in a local approximation does not affect ranking and is removed by query centering in \Cref{sec:objective}. The approximation is empirical and need not describe every possible shallow state.

\subsection{Centered Ridge in Closed Form}
\label{sec:objective}

Only the order within a candidate set matters, and Equation \ref{eq:fusion} is unchanged when a constant is added to every score for one query. We give the fit the same invariance by removing the query mean from states and targets. Let $J \in \mathbb{R}^{K \times K}$ be the matrix of all ones and

\begin{equation}
  P \;=\; I - \tfrac{1}{K} J
  \label{eq:centering}
\end{equation}

the projector that removes the mean, with $I$ the identity of the indicated size. Stacking the states of a query into $H_q \in \mathbb{R}^{K \times d}$ and its targets into $t_q \in \mathbb{R}^{K}$, the readout solves

\begin{equation}
  \min_{a \in \mathbb{R}^{d}}\;
  \sum_{q \in \mathcal{Q}_{\mathrm{fit}}}
  \big\lVert P\big(H_q a - t_q\big) \big\rVert_2^{2}
  \;+\; \lambda \lVert a \rVert_{2}^{2},
  \label{eq:objective}
\end{equation}

over a fixed calibration set $\mathcal{Q}_{\mathrm{fit}}$, with $\lambda > 0$ the ridge strength. Since $P$ is idempotent and symmetric, $H_q^{\top} P H_q = (P H_q)^{\top} (P H_q)$, so Equation \ref{eq:objective} is an ordinary ridge problem on the centered data and carries no intercept. Collecting the centered states of every query into $H_c$ and the centered targets into $t_c$, its solution is

\begin{equation}
  a \;=\; \big( H_c^{\top} H_c + \lambda I \big)^{-1} H_c^{\top} t_c ,
  \label{eq:solution}
\end{equation}

obtained by one linear solve. This fit requires no gradient step and leaves the backbone unchanged. \Cref{alg:method} states the two procedures in full. Although Equation \ref{eq:objective} uses no relevance judgement, the following proposition relates agreement with its teacher target to the relevance correlation $c_s$ in Proposition \ref{prop:weight}.

\begin{proposition}
\label{prop:bound}
Fix a candidate set and let $c_t = \tfrac{1}{K}\langle z(\mathbf{t}), \mathbf{y}\rangle$ be the correlation of the target with relevance and $r = \tfrac{1}{K}\langle z(\mathbf{s}), z(\mathbf{t})\rangle$ the correlation of the model score with the target. Then
\begin{equation}
  c_s \;\ge\; r\, c_t \;-\; \sqrt{\big(1-r^{2}\big)\big(1-c_t^{2}\big)},
  \label{eq:bound}
\end{equation}
and for $r, c_t \in (0,1)$ the right hand side is strictly increasing in $r$.
\end{proposition}

\begin{proof}[Proof sketch]
The correlation matrix of $z(\mathbf{s})$, $z(\mathbf{t})$ and $\mathbf{y}$ is a Gram matrix and hence positive semidefinite, so its determinant $1 - r^{2} - c_t^{2} - c_s^{2} + 2\, r\, c_t\, c_s$ is at least zero. Read as a quadratic in $c_s$, this confines $c_s$ between the roots $r\, c_t \pm \sqrt{(1-r^{2})(1-c_t^{2})}$, and the lower root is Equation \ref{eq:bound}. Its derivative in $r$ is $c_t + r\sqrt{(1-c_t^{2})/(1-r^{2})}$, which is positive on $(0,1)$. \Cref{sec:proof2} gives the full proof.
\end{proof}

The objective fits the teacher margin without reading a relevance judgement. Once the backbone is fixed, $c_t$ is fixed as well, while agreement $r$ is the term in Equation \ref{eq:bound} that the readout can affect. For any candidate set with $r,c_t \in (0,1)$, increasing $r$ raises the lower bound on $c_s$. At $\lambda \to 0$, the solution of Equation \ref{eq:objective} maximizes squared correlation with the teacher over the centered calibration data within the linear family. This is a global property of the fit and does not guarantee an increase in $r$ or $c_s$ for every candidate set. Positive $\lambda$ limits the size of the fitted coefficients and can reduce variance when calibration data are limited. We select $\lambda$ by minimizing the centered squared error of Equation \ref{eq:objective}, without the penalty, on held out queries. This uses the target of Equation \ref{eq:teacher} without relevance annotation (\Cref{sec:lambda}).

\begin{algorithm}[t]
\caption{\method.}
\label{alg:method}
\begin{algorithmic}[1]
\Require calibration queries $\mathcal{Q}_{\mathrm{fit}}$, depth $\ell$, ridge $\lambda$, weight $w$
\Statex
\Function{Fit}{$\mathcal{Q}_{\mathrm{fit}}$}\Comment{once, no relevance judgement}
  \ForAll{$q \in \mathcal{Q}_{\mathrm{fit}},\; p \in \mathcal{C}_q$}
    \State $h_{q,p} \gets \textsc{Shallow}(q, p, \ell)$
    \State $t_{q,p} \gets o_{\texttt{yes}} - o_{\texttt{no}}$ \textbf{from} $\textsc{Forward}(q, p,\, \mathrm{depth}{=}L,\, \mathrm{uncompressed})$
  \EndFor
  \State $H_c, t_c \gets P H,\; P t$ \Comment{$P = I - \tfrac{1}{K} J$, per candidate set}
  \State \Return $a \gets (H_c^{\top} H_c + \lambda I)^{-1} H_c^{\top} t_c$
\EndFunction
\Statex
\Function{Rank}{$q$}\Comment{online, per request}
  \State $\mathcal{C}_q,\, \mathbf{b}_q \gets \textsc{Retrieve}(q)$
  \ForAll{$p \in \mathcal{C}_q$}
    \State $s_{q,p} \gets a^{\top}\, \textsc{Shallow}(q, p, \ell)$
  \EndFor
  \State \Return $\textsc{ArgSort}\big((1-w)\, z(\mathbf{b}_q) + w\, z(\mathbf{s}_q)\big)$
\EndFunction
\end{algorithmic}
\end{algorithm}

\section{Evaluation}
\label{sec:eval}

\Cref{sec:quality} reports ranking quality, \Cref{sec:tradeoff} the accuracy and latency tradeoff, \Cref{sec:operating} the validation of the two propositions, \Cref{sec:ablation} the component ablation, and \Cref{sec:robustness} the parameter sensitivity.

\textbf{Datasets and models.} We evaluate on ViDoRe 2 \citep{vidore2} and ViDoRe 3 \citep{vidore3}, 1038 and 4730 queries over four and eight corpora, on one NVIDIA A100 with 40 GiB. ColNomic-7B \citep{colnomic} and ColPali-v1.3 \citep{colpali} each supply twenty candidates, and Qwen2.5-VL-7B \citep{qwen25vl} and GLM-4.1V-9B \citep{glm41v} each supply a readout, taken at layer 21 from one third of the visual positions.

\textbf{Baselines.} We compare against the retriever alone, against Guided Query Refinement \citep{gqr} under the retriever its published configuration uses, and against a full cross encoder, which serves as a full model reference, not a competing baseline. Two hybrid baselines, RRF \citep{rrf} and Score Aggregation \citep{fusionfunctions}, fuse the retriever with the text retriever Linq-Embed \citep{linq}, Qwen3-Reranker-0.6B \citep{qwen3emb} reranks the page text and MonoQwen2-VL \citep{monoqwen} the page image, both trained on relevance judgements, and the weight of every fused baseline is picked by NDCG@5 on the same inner corpus that sets ours. The reference is the full Qwen 7B scoring every candidate image; both backbones are compared against it.

For each backbone, we fit four readouts using the 154 calibration queries sampled from the four ViDoRe 2 corpora under leave one corpus out. In each fold, queries from three corpora are used for fitting, and the resulting readout evaluates the remaining corpus. On ViDoRe 3, each corpus is evaluated with all four readouts, and we report the mean of the resulting metrics. We assign one fusion weight per corpus using the three correlations defined in Proposition \ref{prop:weight}, estimated on the next corpus in a fixed cyclic order. Weight estimation is performed separately within ViDoRe 2 and ViDoRe 3.

\subsection{Ranking Quality}
\label{sec:quality}

Table \ref{tab:ndcg5_vidore2} and Table \ref{tab:ndcg5_vidore3} give NDCG@5 by subset. On ViDoRe 2 with ColNomic candidates the retriever scores 60.9 and the full cross encoder 72.4, leaving 11.5 points of distance. GQR closes 2.9 of them, RRF, Score Aggregation and Qwen3-Reranker-0.6B reach 62.3, 62.2 and 61.6, and MonoQwen2-VL reaches 65.4. The frozen readout closes 8.1 with the Qwen backbone and 10.3 with the GLM backbone, which is 70 and 89 percent of what the cross encoder is worth.

NDCG@5 is the primary metric throughout. Table \ref{tab:metrics} in Appendix \ref{sec:extra} reports four further metrics for every arm, and \method{} leads every external baseline under all of them. The readout improves on the retriever in every one of the four settings, by 8.1 and 10.3 points on ViDoRe 2 with ColNomic candidates, by 14.7 and 15.8 with ColPali, and by 3.4 and 3.0 on ViDoRe 3.

The same four vectors serve both retrievers and both benchmarks. They were fitted on ColNomic candidates over four ViDoRe 2 corpora, so the ColPali pools and every ViDoRe 3 corpus are outside the fit entirely, and the coefficients still close 85 and 91 percent of the ColPali distance on ViDoRe 2 and 82 percent of it on ViDoRe 3 with either backbone.

The primary comparison is against the full cross encoder reference. On ViDoRe 2 the GLM readout sits 1.2 points below the full cross encoder with ColNomic candidates and 1.5 points below it with ColPali candidates, so a vector fitted once on 154 queries comes within two points of the reference cross encoder, at a small fraction of its cost.

\begin{table}[t]
\centering
\small
\setlength{\tabcolsep}{3.5pt}
\caption{NDCG@5 (\%) on all four subsets of ViDoRe 2. Within each block the deltas are taken against the retriever alone. The full cross encoder scores every candidate from its image and is reported as a reference, not as a competing baseline. GQR is listed only under the retriever its published configuration uses.}
\label{tab:ndcg5_vidore2}
\resizebox{\textwidth}{!}{%
\begin{tabular}{llrlrlrlrlrl}
\toprule
\multirow{2}{*}{Retriever} & \multirow{2}{*}{Reranker} & \multicolumn{2}{c}{Avg} & \multicolumn{2}{c}{Biomed} & \multicolumn{2}{c}{Econ} & \multicolumn{2}{c}{ESG-H} & \multicolumn{2}{c}{ESG-F} \\
\cmidrule(lr){3-12}
 &  & val & $\Delta$ & val & $\Delta$ & val & $\Delta$ & val & $\Delta$ & val & $\Delta$ \\
\midrule
\multirow{9}{*}{ColNomic-7B} & No reranking & 60.9 & {\scriptsize --} & 63.9 & {\scriptsize --} & 55.1 & {\scriptsize --} & 70.7 & {\scriptsize --} & 53.8 & {\scriptsize --} \\
 & Full cross encoder & 72.4 & {\scriptsize $+$11.5} & 71.7 & {\scriptsize $+$7.8} & 68.2 & {\scriptsize $+$13.1} & 83.1 & {\scriptsize $+$12.4} & 66.7 & {\scriptsize $+$12.9} \\
 & RRF & 62.3 & {\scriptsize $+$1.4} & 63.1 & {\scriptsize $-$0.8} & 58.7 & {\scriptsize $+$3.6} & 68.2 & {\scriptsize $-$2.6} & 59.1 & {\scriptsize $+$5.3} \\
 & Score Aggregation & 62.2 & {\scriptsize $+$1.3} & 63.6 & {\scriptsize $-$0.2} & 57.4 & {\scriptsize $+$2.2} & 69.9 & {\scriptsize $-$0.9} & 58.0 & {\scriptsize $+$4.2} \\
 & GQR & 63.8 & {\scriptsize $+$2.9} & 66.3 & {\scriptsize $+$2.4} & 59.4 & {\scriptsize $+$4.2} & 69.3 & {\scriptsize $-$1.4} & 60.1 & {\scriptsize $+$6.3} \\
 & Qwen3-Reranker-0.6B & 61.6 & {\scriptsize $+$0.7} & 64.3 & {\scriptsize $+$0.4} & 57.1 & {\scriptsize $+$2.0} & 68.3 & {\scriptsize $-$2.5} & 56.9 & {\scriptsize $+$3.1} \\
 & MonoQwen2-VL & 65.4 & {\scriptsize $+$4.5} & 67.6 & {\scriptsize $+$3.7} & 59.3 & {\scriptsize $+$4.1} & 73.7 & {\scriptsize $+$2.9} & 61.0 & {\scriptsize $+$7.2} \\
\cmidrule(lr){2-12}
 & Ours (Qwen2.5-VL-7B) & 69.0 & {\scriptsize\gain{8.1}} & 71.4 & {\scriptsize\gain{7.5}} & 64.8 & {\scriptsize\gain{9.6}} & 76.9 & {\scriptsize\gain{6.1}} & 63.1 & {\scriptsize\gain{9.3}} \\
 & Ours (GLM-4.1V-9B) & 71.2 & {\scriptsize\gain{10.3}} & 72.1 & {\scriptsize\gain{8.3}} & 65.8 & {\scriptsize\gain{10.7}} & 79.1 & {\scriptsize\gain{8.4}} & 67.8 & {\scriptsize\gain{14.0}} \\
\midrule
\multirow{8}{*}{ColPali-v1.3} & No reranking & 50.2 & {\scriptsize --} & 56.2 & {\scriptsize --} & 46.3 & {\scriptsize --} & 49.7 & {\scriptsize --} & 48.8 & {\scriptsize --} \\
 & Full cross encoder & 67.6 & {\scriptsize $+$17.3} & 68.8 & {\scriptsize $+$12.7} & 64.1 & {\scriptsize $+$17.9} & 70.0 & {\scriptsize $+$20.3} & 67.4 & {\scriptsize $+$18.5} \\
 & RRF & 55.8 & {\scriptsize $+$5.5} & 59.9 & {\scriptsize $+$3.7} & 53.5 & {\scriptsize $+$7.2} & 56.8 & {\scriptsize $+$7.1} & 52.9 & {\scriptsize $+$4.0} \\
 & Score Aggregation & 56.8 & {\scriptsize $+$6.6} & 60.2 & {\scriptsize $+$4.1} & 54.9 & {\scriptsize $+$8.6} & 57.2 & {\scriptsize $+$7.5} & 54.9 & {\scriptsize $+$6.1} \\
 & Qwen3-Reranker-0.6B & 54.7 & {\scriptsize $+$4.5} & 56.5 & {\scriptsize $+$0.4} & 51.8 & {\scriptsize $+$5.6} & 57.8 & {\scriptsize $+$8.0} & 52.7 & {\scriptsize $+$3.8} \\
 & MonoQwen2-VL & 62.6 & {\scriptsize $+$12.4} & 64.8 & {\scriptsize $+$8.6} & 56.1 & {\scriptsize $+$9.9} & 67.3 & {\scriptsize $+$17.6} & 62.3 & {\scriptsize $+$13.5} \\
\cmidrule(lr){2-12}
 & Ours (Qwen2.5-VL-7B) & 64.9 & {\scriptsize\gain{14.7}} & 68.3 & {\scriptsize\gain{12.2}} & 62.4 & {\scriptsize\gain{16.2}} & 67.4 & {\scriptsize\gain{17.6}} & 62.6 & {\scriptsize\gain{13.8}} \\
 & Ours (GLM-4.1V-9B) & 66.1 & {\scriptsize\gain{15.8}} & 69.1 & {\scriptsize\gain{13.0}} & 62.9 & {\scriptsize\gain{16.7}} & 68.2 & {\scriptsize\gain{18.5}} & 64.1 & {\scriptsize\gain{15.3}} \\
\bottomrule
\end{tabular}}
\end{table}

\begin{table}[t]
\centering
\small
\setlength{\tabcolsep}{3.5pt}
\caption{NDCG@5 (\%) on all eight public subsets of ViDoRe 3. Within each block the deltas are taken against the retriever alone. The full cross encoder scores every candidate from its image and is reported as a reference, not as a competing baseline. GQR is listed only under the retriever its published configuration uses.}
\label{tab:ndcg5_vidore3}
\resizebox{\textwidth}{!}{%
\begin{tabular}{llrlrlrlrlrlrlrlrlrl}
\toprule
\multirow{2}{*}{Retriever} & \multirow{2}{*}{Reranker} & \multicolumn{2}{c}{Avg} & \multicolumn{2}{c}{cs} & \multicolumn{2}{c}{energy} & \multicolumn{2}{c}{fin en} & \multicolumn{2}{c}{fin fr} & \multicolumn{2}{c}{hr} & \multicolumn{2}{c}{indus} & \multicolumn{2}{c}{pharma} & \multicolumn{2}{c}{physics} \\
\cmidrule(lr){3-20}
 &  & val & $\Delta$ & val & $\Delta$ & val & $\Delta$ & val & $\Delta$ & val & $\Delta$ & val & $\Delta$ & val & $\Delta$ & val & $\Delta$ & val & $\Delta$ \\
\midrule
\multirow{9}{*}{ColNomic-7B} & No reranking & 52.8 & {\scriptsize --} & 72.4 & {\scriptsize --} & 61.9 & {\scriptsize --} & 43.4 & {\scriptsize --} & 42.1 & {\scriptsize --} & 54.1 & {\scriptsize --} & 43.7 & {\scriptsize --} & 57.8 & {\scriptsize --} & 47.1 & {\scriptsize --} \\
 & Full cross encoder & 58.0 & {\scriptsize $+$5.2} & 76.6 & {\scriptsize $+$4.1} & 67.5 & {\scriptsize $+$5.7} & 52.1 & {\scriptsize $+$8.6} & 48.6 & {\scriptsize $+$6.6} & 59.4 & {\scriptsize $+$5.3} & 50.0 & {\scriptsize $+$6.3} & 61.9 & {\scriptsize $+$4.0} & 48.2 & {\scriptsize $+$1.1} \\
 & RRF & 51.4 & {\scriptsize $-$1.4} & 70.6 & {\scriptsize $-$1.9} & 57.4 & {\scriptsize $-$4.4} & 41.1 & {\scriptsize $-$2.3} & 42.6 & {\scriptsize $+$0.5} & 51.2 & {\scriptsize $-$2.9} & 42.2 & {\scriptsize $-$1.5} & 58.9 & {\scriptsize $+$1.1} & 47.4 & {\scriptsize $+$0.3} \\
 & Score Aggregation & 53.5 & {\scriptsize $+$0.7} & 72.4 & {\scriptsize $+$0.0} & 62.2 & {\scriptsize $+$0.3} & 44.2 & {\scriptsize $+$0.8} & 43.0 & {\scriptsize $+$0.9} & 54.8 & {\scriptsize $+$0.7} & 44.5 & {\scriptsize $+$0.8} & 59.6 & {\scriptsize $+$1.7} & 47.2 & {\scriptsize $+$0.1} \\
 & GQR & 54.3 & {\scriptsize $+$1.5} & 73.4 & {\scriptsize $+$1.0} & 63.4 & {\scriptsize $+$1.5} & 45.6 & {\scriptsize $+$2.2} & 43.4 & {\scriptsize $+$1.3} & 55.2 & {\scriptsize $+$1.1} & 45.1 & {\scriptsize $+$1.5} & 59.9 & {\scriptsize $+$2.1} & 48.1 & {\scriptsize $+$1.0} \\
 & Qwen3-Reranker-0.6B & 54.0 & {\scriptsize $+$1.2} & 68.8 & {\scriptsize $-$3.6} & 64.0 & {\scriptsize $+$2.1} & 45.7 & {\scriptsize $+$2.3} & 44.2 & {\scriptsize $+$2.2} & 56.5 & {\scriptsize $+$2.4} & 45.0 & {\scriptsize $+$1.3} & 60.2 & {\scriptsize $+$2.4} & 47.3 & {\scriptsize $+$0.2} \\
 & MonoQwen2-VL & 54.5 & {\scriptsize $+$1.7} & 73.9 & {\scriptsize $+$1.5} & 63.1 & {\scriptsize $+$1.3} & 47.4 & {\scriptsize $+$4.0} & 43.8 & {\scriptsize $+$1.7} & 55.4 & {\scriptsize $+$1.3} & 45.9 & {\scriptsize $+$2.2} & 59.1 & {\scriptsize $+$1.3} & 47.7 & {\scriptsize $+$0.6} \\
\cmidrule(lr){2-20}
 & Ours (Qwen2.5-VL-7B) & 56.2 & {\scriptsize\gain{3.4}} & 76.3 & {\scriptsize\gain{3.9}} & 66.2 & {\scriptsize\gain{4.4}} & 48.7 & {\scriptsize\gain{5.2}} & 45.1 & {\scriptsize\gain{3.1}} & 57.6 & {\scriptsize\gain{3.6}} & 47.8 & {\scriptsize\gain{4.1}} & 61.2 & {\scriptsize\gain{3.4}} & 46.9 & {\scriptsize\loss{0.2}} \\
 & Ours (GLM-4.1V-9B) & 55.8 & {\scriptsize\gain{3.0}} & 76.4 & {\scriptsize\gain{4.0}} & 65.2 & {\scriptsize\gain{3.3}} & 47.9 & {\scriptsize\gain{4.5}} & 43.8 & {\scriptsize\gain{1.8}} & 57.1 & {\scriptsize\gain{3.0}} & 46.8 & {\scriptsize\gain{3.1}} & 62.8 & {\scriptsize\gain{5.0}} & 47.7 & {\scriptsize\gain{0.6}} \\
\midrule
\multirow{8}{*}{ColPali-v1.3} & No reranking & 38.8 & {\scriptsize --} & 57.3 & {\scriptsize --} & 43.9 & {\scriptsize --} & 27.3 & {\scriptsize --} & 23.6 & {\scriptsize --} & 40.1 & {\scriptsize --} & 30.1 & {\scriptsize --} & 48.0 & {\scriptsize --} & 39.8 & {\scriptsize --} \\
 & Full cross encoder & 51.1 & {\scriptsize $+$12.3} & 73.5 & {\scriptsize $+$16.2} & 58.3 & {\scriptsize $+$14.4} & 39.8 & {\scriptsize $+$12.4} & 37.7 & {\scriptsize $+$14.1} & 53.6 & {\scriptsize $+$13.4} & 42.0 & {\scriptsize $+$11.9} & 57.5 & {\scriptsize $+$9.6} & 46.1 & {\scriptsize $+$6.2} \\
 & RRF & 42.9 & {\scriptsize $+$4.2} & 60.2 & {\scriptsize $+$2.8} & 48.5 & {\scriptsize $+$4.5} & 29.5 & {\scriptsize $+$2.2} & 28.9 & {\scriptsize $+$5.3} & 45.6 & {\scriptsize $+$5.5} & 32.5 & {\scriptsize $+$2.4} & 54.2 & {\scriptsize $+$6.3} & 44.3 & {\scriptsize $+$4.4} \\
 & Score Aggregation & 44.6 & {\scriptsize $+$5.8} & 64.5 & {\scriptsize $+$7.2} & 51.2 & {\scriptsize $+$7.2} & 31.3 & {\scriptsize $+$4.0} & 31.0 & {\scriptsize $+$7.4} & 45.1 & {\scriptsize $+$5.0} & 32.8 & {\scriptsize $+$2.7} & 55.9 & {\scriptsize $+$8.0} & 44.9 & {\scriptsize $+$5.0} \\
 & Qwen3-Reranker-0.6B & 47.8 & {\scriptsize $+$9.0} & 68.0 & {\scriptsize $+$10.6} & 56.8 & {\scriptsize $+$12.9} & 35.7 & {\scriptsize $+$8.3} & 33.7 & {\scriptsize $+$10.2} & 49.8 & {\scriptsize $+$9.7} & 36.9 & {\scriptsize $+$6.8} & 54.6 & {\scriptsize $+$6.7} & 46.6 & {\scriptsize $+$6.8} \\
 & MonoQwen2-VL & 46.7 & {\scriptsize $+$7.9} & 67.2 & {\scriptsize $+$9.8} & 54.6 & {\scriptsize $+$10.6} & 34.3 & {\scriptsize $+$7.0} & 32.0 & {\scriptsize $+$8.4} & 48.3 & {\scriptsize $+$8.2} & 38.7 & {\scriptsize $+$8.6} & 54.9 & {\scriptsize $+$7.0} & 43.7 & {\scriptsize $+$3.9} \\
\cmidrule(lr){2-20}
 & Ours (Qwen2.5-VL-7B) & 48.8 & {\scriptsize\gain{10.0}} & 71.8 & {\scriptsize\gain{14.4}} & 56.8 & {\scriptsize\gain{12.9}} & 36.5 & {\scriptsize\gain{9.2}} & 34.5 & {\scriptsize\gain{10.9}} & 51.6 & {\scriptsize\gain{11.5}} & 39.2 & {\scriptsize\gain{9.1}} & 57.5 & {\scriptsize\gain{9.5}} & 45.6 & {\scriptsize\gain{5.8}} \\
 & Ours (GLM-4.1V-9B) & 48.8 & {\scriptsize\gain{10.1}} & 75.1 & {\scriptsize\gain{17.7}} & 55.8 & {\scriptsize\gain{11.8}} & 35.0 & {\scriptsize\gain{7.6}} & 32.8 & {\scriptsize\gain{9.2}} & 50.9 & {\scriptsize\gain{10.8}} & 37.5 & {\scriptsize\gain{7.4}} & 57.8 & {\scriptsize\gain{9.8}} & 46.8 & {\scriptsize\gain{7.0}} \\
\bottomrule
\end{tabular}}
\end{table}

\subsection{Accuracy and Latency Tradeoff}
\label{sec:tradeoff}

Figure \ref{fig:frontier} places accuracy against measured cost, and Table \ref{tab:efficiency_vidore2} in Appendix \ref{sec:extra} gives the same numbers for both retrievers. Every latency is the mean over the whole query set and not over a sample. On ViDoRe 2 the retriever costs 60 ms, GQR costs 130 ms, Qwen3-Reranker-0.6B costs 344 ms, MonoQwen2-VL costs 2860 ms, our method costs 156 ms with the Qwen backbone and 141 ms with the GLM backbone, and the full cross encoder costs 6765 ms. The GLM configuration is therefore 48 times cheaper than the cross encoder it matches, and 11 ms more expensive than GQR while closing 65 more percentage points of the available distance.

The curve of Figure \ref{fig:frontier} is steep up to about 150 ms. The first 80 ms of reranking buys 10.3 points and the remaining 6.6 seconds buy 1.2 more. Both backbones sit at the corner, which is the part of the curve a deployment would choose from. On ViDoRe 3 the same ordering holds against the cross encoder, at 213 ms and 233 ms against 7534 ms, but the margin over GQR narrows to 1.6 points.

Figure \ref{fig:breakdown} splits each cost into retrieval, reading what a method keeps for each page, and model computation. \method{} computes for 50 ms with the GLM backbone and 56 ms with the Qwen backbone, against 282 ms for Qwen3-Reranker-0.6B over the page text, 2.6 s for MonoQwen2-VL over the page images and 6.7 s for the full cross encoder.

\begin{figure}[t]
\centering
\sloppy
\begin{minipage}[t]{0.315\textwidth}
  \centering
  \includegraphics[width=\linewidth]{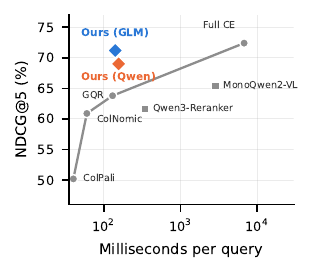}
  \caption{Pareto frontier.}
  \label{fig:frontier}
\end{minipage}\hfill
\begin{minipage}[t]{0.315\textwidth}
  \centering
  \includegraphics[width=\linewidth]{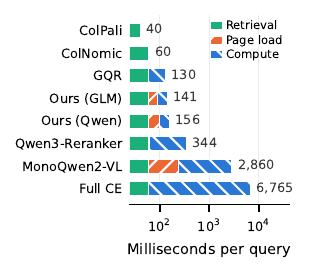}
  \caption{Latency breakdown.}
  \label{fig:breakdown}
\end{minipage}\hfill
\begin{minipage}[t]{0.315\textwidth}
  \centering
  \includegraphics[width=\linewidth]{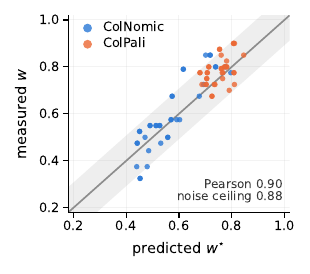}
  \caption{Predicted weight. The band is the 95\% interval of the measured weight.}
  \label{fig:weight}
\end{minipage}

\end{figure}

\subsection{Theoretical Validation}
\label{sec:operating}

We evaluate Proposition \ref{prop:weight} on 72 subsets: the full cross encoder and both readouts on every corpus under both retrievers. Both $c_b-\rho c_s$ and $c_s-\rho c_b$ are at least 0.085, giving interior optima. We define the empirical optimum as the mean of weights within one NDCG@5 point of the best. Over 200 random half splits, Spearman--Brown correction gives reliabilities of 0.903 and 0.863 for the predicted and measured weights. Their geometric mean, 0.883, estimates a correlation ceiling under parallel measurements and independent errors. Shared annotations violate error independence here, so this ceiling does not apply. \Cref{fig:weight} reports a correlation of 0.898, with the weaker retriever leaving more weight to the reranker. Here the correlations are read on each subset itself, which tests the proposition, while the deployed weight reads them on another corpus.

Proposition \ref{prop:bound} ties the agreement Equation \ref{eq:objective} maximizes without labels to $c_s$. We test the tie by holding the corpus fixed, which leaves $c_t$ constant, and moving only the fit over every layer crossed with every regularization strength. \Cref{fig:chain} shows $r$ and $c_s$ correlating from 0.894 to 0.994 over that family, with a median of 0.975 on the eight pairings of backbone and corpus.

\subsection{Component Ablation}
\label{sec:ablation}

Table \ref{tab:ablation} reports component ablations for Qwen2.5-VL-7B on the calibration sample. Each row removes one component with the remaining settings fixed.

The hardest row is the logit lens, which reads the same cached state through the backbone's own unembedding and costs nothing to fit. The solve is worth 4.1 points over it. Consuming the model score on its own costs 2.7 points, which is Proposition \ref{prop:weight} measured. Centering inside the candidate list is worth 0.6 points, which is what Equation \ref{eq:objective} buys over a pointwise regression with an intercept.

\begin{figure}[t]
\centering
\begin{minipage}[t]{0.315\textwidth}
  \centering
  \includegraphics[width=\linewidth]{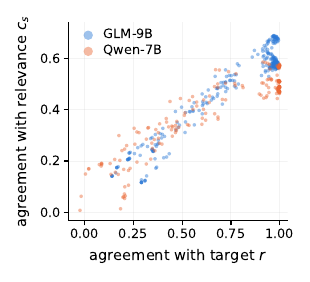}
  \caption{Objective alignment.}
  \label{fig:chain}
\end{minipage}\hfill
\begin{minipage}[t]{0.315\textwidth}
  \centering
  \includegraphics[width=\linewidth]{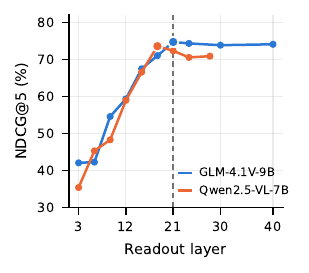}
  \caption{Accuracy by layer.}
  \label{fig:depth}
\end{minipage}\hfill
\begin{minipage}[t]{0.315\textwidth}
  \centering
  \includegraphics[width=\linewidth]{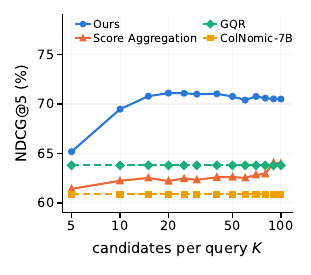}
  \caption{Pool depth.}
  \label{fig:pooldepth}
\end{minipage}
\end{figure}

\begin{table}[t]
\centering
\small
\setlength{\tabcolsep}{2pt}
\caption{Component ablation on the 154 calibration queries under leave one corpus out, with $w=0.85$ unless fusion is removed.}
\label{tab:ablation}
\begin{tabular*}{\textwidth}{@{\extracolsep{\fill}}lrlrlrlrlrl}
\toprule
\multirow{2}{*}{Component} & \multicolumn{2}{c}{Avg} & \multicolumn{2}{c}{Biomed} & \multicolumn{2}{c}{Econ} & \multicolumn{2}{c}{ESG-H} & \multicolumn{2}{c}{ESG-F} \\
\cmidrule(lr){2-11}
 & val & $\Delta$ & val & $\Delta$ & val & $\Delta$ & val & $\Delta$ & val & $\Delta$ \\
\midrule
($-$) Operating point & 69.6 & \loss{2.7} & 54.8 & \loss{4.7} & 71.1 & \loss{2.7} & 75.5 & \loss{3.7} & 76.0 & \loss{0.8} \\
($-$) Fitted readout & 68.3 & \loss{4.1} & 55.4 & \loss{4.0} & 71.5 & \loss{2.3} & 77.5 & \loss{1.7} & 68.6 & \loss{8.3} \\
($-$) List centering & 71.7 & \loss{0.6} & 59.1 & \loss{0.4} & 73.0 & \loss{0.9} & 79.0 & \loss{0.2} & 75.9 & \loss{0.9} \\
\midrule
\textbf{\method} & \textbf{72.3} & -- & \textbf{59.4} & -- & \textbf{73.8} & -- & \textbf{79.2} & -- & \textbf{76.8} & -- \\
\bottomrule
\end{tabular*}
\end{table}

\subsection{Parameter Sensitivity}
\label{sec:robustness}

Every input the method takes sits on a plateau. \Cref{fig:depth} reads the state at each depth of the backbone and finds both peaking before their last layer. Stopping at an intermediate layer follows prior work \citep{prettr,deformer,minireranker} and is not part of our contribution, and we fix $\ell = 21$ in advance because choosing the peak of this curve would read relevance judgements. \Cref{fig:pooldepth} scores pools from five to a hundred candidates with the coefficients frozen, which climbs to the deployed twenty and is flat over the fivefold range above it. The retriever's line there is flat by construction, since the top five of the top twenty is the top five of the top hundred. The ridge strength and the fusion weight behave the same way, moving ranking quality by less than 1.6 points across two decades of $\lambda$ and by less than half a point across the weights nearest the best.

The fit is also cheap in the only resource it consumes. Drawing calibration queries at random and refitting, sixteen queries already reach 99.4 percent of the quality obtained from all 154 on GLM and 98.3 percent on Qwen.

\section{Conclusion}

Multimodal rerankers improve visual document retrieval but incur substantial inference costs, and recovering accuracy after compression can require relevance judgements. \method{} combines retriever and reranker scores using a fusion weight derived in closed form from their correlations with relevance and with each other. The calibration target is the uncompressed model's answer token margin, computed by an inner product with its normalized final state. A shallow linear readout approximates this target through one centered ridge solve, using differences within candidate lists without gradient updates or relevance annotations. On ViDoRe 2 and ViDoRe 3, its NDCG@5 lies within 1.2\,pp of a full cross encoder at up to 48 times the speed.

\bibliography{references}

@inproceedings{gqr,
  title={Guided Query Refinement: Multimodal Hybrid Retrieval with Test-Time Optimization},
  author={Uzan, Omri and Yehudai, Asaf and Pony, Roi and Shnarch, Eyal and Gera, Ariel},
  booktitle={International Conference on Learning Representations},
  volume={2026},
  pages={495--522},
  year={2026}
}

@inproceedings{colpali,
  title={Colpali: Efficient document retrieval with vision language models},
  author={Faysse, Manuel and Sibille, Hugues and Wu, Tony and Omrani, Bilel and Viaud, Gautier and Hudelot, C{\'e}line and Colombo, Pierre},
  booktitle={International Conference on Learning Representations},
  volume={2025},
  pages={61424--61449},
  year={2025}
}

@inproceedings{colbert,
  title={Colbert: Efficient and effective passage search via contextualized late interaction over bert},
  author={Khattab, Omar and Zaharia, Matei},
  booktitle={Proceedings of the 43rd International ACM SIGIR conference on research and development in Information Retrieval},
  pages={39--48},
  year={2020}
}

@article{vidore2,
  title={Vidore benchmark v2: Raising the bar for visual retrieval},
  author={Mac{\'e}, Quentin and Loison, Ant{\'o}nio and Faysse, Manuel},
  journal={arXiv preprint arXiv:2505.17166},
  year={2025}
}

@article{minireranker,
  title={miniReranker: Efficient Multimodal Reranking through Visual Cache Reuse and Interaction Sparsity},
  author={Fan, Yingqi and Lu, Xuan and Zhao, Anhao and Tong, Junlong and Nie, Ping and Zou, Kai and Ma, Yunpu and Zhang, Wei and Shen, Xiaoyu},
  journal={arXiv preprint arXiv:2606.10759},
  year={2026}
}

@article{rankprune,
  title={From Saliency to Discriminability: Rank-Preserving Visual Token Pruning for VLM Rerankers},
  author={Liu, Siyi and Yang, Hanjun and Zhang, Chenchen and Zhu, Xiaorong and Zuo, Xinyu and Duan, Lisheng and Liang, Haijin and Ma, Jin and Pu, Junfu and Zhang, Yongqi},
  journal={arXiv preprint arXiv:2609.00667},
  year={2026}
}

@inproceedings{prettr,
  title={Efficient document re-ranking for transformers by precomputing term representations},
  author={MacAvaney, Sean and Nardini, Franco Maria and Perego, Raffaele and Tonellotto, Nicola and Goharian, Nazli and Frieder, Ophir},
  booktitle={Proceedings of the 43rd international ACM SIGIR conference on research and development in information retrieval},
  pages={49--58},
  year={2020}
}

@inproceedings{deformer,
  title={DeFormer: Decomposing pre-trained transformers for faster question answering},
  author={Cao, Qingqing and Trivedi, Harsh and Balasubramanian, Aruna and Balasubramanian, Niranjan},
  booktitle={Proceedings of the 58th Annual Meeting of the Association for Computational Linguistics},
  pages={4487--4497},
  year={2020}
}

@inproceedings{deebert,
  title={DeeBERT: Dynamic early exiting for accelerating BERT inference},
  author={Xin, Ji and Tang, Raphael and Lee, Jaejun and Yu, Yaoliang and Lin, Jimmy},
  booktitle={Proceedings of the 58th annual meeting of the association for computational linguistics},
  pages={2246--2251},
  year={2020}
}

@article{calm,
  title={Confident adaptive language modeling},
  author={Schuster, Tal and Fisch, Adam and Gupta, Jai and Dehghani, Mostafa and Bahri, Dara and Tran, Vinh and Tay, Yi and Metzler, Donald},
  journal={Advances in Neural Information Processing Systems},
  volume={35},
  pages={17456--17472},
  year={2022}
}

@inproceedings{see,
  title={Efficient re-ranking with cross-encoders via early exit},
  author={Busolin, Francesco and Lucchese, Claudio and Nardini, Franco Maria and Orlando, Salvatore and Perego, Raffaele and Trani, Salvatore and Veneri, Alberto},
  booktitle={Proceedings of the 48th International ACM SIGIR Conference on Research and Development in Information Retrieval},
  pages={2534--2544},
  year={2025}
}

@inproceedings{e2rank,
  title={E2rank: Efficient and effective layer-wise reranking},
  author={Campagnano, Cesare and Mallia, Antonio and Pertschuk, Jack and Silvestri, Fabrizio},
  booktitle={European Conference on Information Retrieval},
  pages={417--426},
  year={2025},
  organization={Springer}
}

@article{kivi,
  title={Kivi: A tuning-free asymmetric 2bit quantization for kv cache},
  author={Liu, Zirui and Yuan, Jiayi and Jin, Hongye and Zhong, Shaochen and Xu, Zhaozhuo and Braverman, Vladimir and Chen, Beidi and Hu, Xia},
  journal={arXiv preprint arXiv:2402.02750},
  year={2024}
}

@inproceedings{rrf,
  title={Reciprocal rank fusion outperforms condorcet and individual rank learning methods},
  author={Cormack, Gordon V and Clarke, Charles LA and Buettcher, Stefan},
  booktitle={Proceedings of the 32nd international ACM SIGIR conference on Research and development in information retrieval},
  pages={758--759},
  year={2009}
}

@article{fusionfunctions,
  title={An analysis of fusion functions for hybrid retrieval},
  author={Bruch, Sebastian and Gai, Siyu and Ingber, Amir},
  journal={ACM Transactions on Information Systems},
  volume={42},
  number={1},
  pages={1--35},
  year={2023},
  publisher={ACM New York, NY}
}

@inproceedings{wanglin,
  title={Bert-based dense retrievers require interpolation with bm25 for effective passage retrieval},
  author={Wang, Shuai and Zhuang, Shengyao and Zuccon, Guido},
  booktitle={Proceedings of the 2021 ACM SIGIR international conference on theory of information retrieval},
  pages={317--324},
  year={2021}
}

@article{fastforward,
  title={Efficient neural ranking using forward indexes and lightweight encoders},
  author={Leonhardt, Jurek and M{\"u}ller, Henrik and Rudra, Koustav and Khosla, Megha and Anand, Abhijit and Anand, Avishek},
  journal={ACM Transactions on Information Systems},
  volume={42},
  number={5},
  pages={1--34},
  year={2024},
  publisher={ACM New York, NY}
}

@inproceedings{askari,
  title={Injecting the BM25 score as text improves BERT-based re-rankers},
  author={Askari, Arian and Abolghasemi, Amin and Pasi, Gabriella and Kraaij, Wessel and Verberne, Suzan},
  booktitle={European Conference on Information Retrieval},
  pages={66--83},
  year={2023},
  organization={Springer}
}

@article{tunedlens,
  title={Eliciting latent predictions from transformers with the tuned lens},
  author={Belrose, Nora and Ostrovsky, Igor and McKinney, Lev and Furman, Zach and Smith, Logan and Halawi, Danny and Biderman, Stella and Steinhardt, Jacob},
  journal={arXiv preprint arXiv:2303.08112},
  year={2023}
}

@inproceedings{byot,
  title={Be your own teacher: Improve the performance of convolutional neural networks via self distillation},
  author={Zhang, Linfeng and Song, Jiebo and Gao, Anni and Chen, Jingwei and Bao, Chenglong and Ma, Kaisheng},
  booktitle={2019 IEEE/CVF International Conference on Computer Vision (ICCV)},
  pages={3712--3721},
  year={2019},
  organization={IEEE}
}

@inproceedings{fastbert,
  title={Fastbert: a self-distilling bert with adaptive inference time},
  author={Liu, Weijie and Zhou, Peng and Wang, Zhiruo and Zhao, Zhe and Deng, Haotang and Ju, Qi},
  booktitle={Proceedings of the 58th annual meeting of the association for computational linguistics},
  pages={6035--6044},
  year={2020}
}

@article{refit,
  title={ReFIT: Relevance feedback from a reranker during inference},
  author={Reddy, Revanth Gangi and Dasigi, Pradeep and Sultan, Md Arafat and Cohan, Arman and Sil, Avirup and Ji, Heng and Hajishirzi, Hannaneh},
  journal={arXiv preprint arXiv:2305.11744},
  year={2023}
}

@inproceedings{tour,
  title={Optimizing test-time query representations for dense retrieval},
  author={Sung, Mujeen and Park, Jungsoo and Kang, Jaewoo and Chen, Danqi and Lee, Jinhyuk},
  booktitle={Findings of the Association for Computational Linguistics: ACL 2023},
  pages={5731--5746},
  year={2023}
}

@article{ziprerank,
  title={Very efficient listwise multimodal reranking for long documents},
  author={Sun, Yiqun and Wei, Pengfei and Hsieh, Lawrence B},
  journal={arXiv preprint arXiv:2605.11864},
  year={2026}
}

@article{miner,
  title={MINER: Mining Multimodal Internal Representation for Efficient Retrieval},
  author={Li, Weien and Song, Rui and Li, Zeyu and Liu, Haochen and Zhang, Gonghao and Jiao, Difan and Tang, Zhenwei and He, Bowei and Wu, Haolun and Liu, Xue and others},
  journal={arXiv preprint arXiv:2605.06460},
  year={2026}
}

@article{qwen3vlrerank,
  title={Qwen3-vl-embedding and qwen3-vl-reranker: A unified framework for state-of-the-art multimodal retrieval and ranking},
  author={Li, Mingxin and Zhang, Yanzhao and Long, Dingkun and Chen, Keqin and Song, Sibo and Bai, Shuai and Yang, Zhibo and Xie, Pengjun and Yang, An and Liu, Dayiheng and others},
  journal={arXiv preprint arXiv:2601.04720},
  year={2026}
}

@inproceedings{vidore3,
  title={Vidore v3: A comprehensive evaluation of retrieval augmented generation in complex real-world scenarios},
  author={Loison, Ant{\'o}nio and Mac{\'e}, Quentin and Edy, Antoine and Xing, Victor and Balough, Tom and Moreira, Gabriel de Souza P and Liu, Bo and Faysse, Manuel and Hudelot, C{\'e}line and Viaud, Gautier},
  booktitle={Proceedings of the 64th Annual Meeting of the Association for Computational Linguistics (Volume 1: Long Papers)},
  pages={16570--16600},
  year={2026}
}

@article{nanovdr,
  title={NanoVDR: Distilling a 2B Vision-Language Retriever into a 70M Text-Only Encoder for Visual Document Retrieval},
  author={Liu, Zhuchenyang and Zhang, Yao and Xiao, Yu},
  journal={arXiv preprint arXiv:2603.12824},
  year={2026}
}

@article{glm41v,
  title={Glm-4.5 v and glm-4.1 v-thinking: Towards versatile multimodal reasoning with scalable reinforcement learning},
  author={Hong, Wenyi and Yu, Wenmeng and Gu, Xiaotao and Wang, Guo and Gan, Guobing and Tang, Haomiao and Cheng, Jiale and Qi, Ji and Ji, Junhui and Pan, Lihang and others},
  journal={arXiv preprint arXiv:2507.01006},
  year={2025}
}

@misc{colnomic,
  author       = {{NomicAI}},
  title        = {Nomic Embed Multimodal: Interleaved Text, Image, and Screenshots
                  for Visual Document Retrieval},
  year         = {2025},
  url          = {https://nomic.ai/blog/posts/nomic-embed-multimodal}
}

@misc{qwen25vl,
      title={Qwen2.5-VL Technical Report}, 
      author={Shuai Bai and Keqin Chen and Xuejing Liu and Jialin Wang and Wenbin Ge and Sibo Song and Kai Dang and Peng Wang and Shijie Wang and Jun Tang and Humen Zhong and Yuanzhi Zhu and Mingkun Yang and Zhaohai Li and Jianqiang Wan and Pengfei Wang and Wei Ding and Zheren Fu and Yiheng Xu and Jiabo Ye and Xi Zhang and Tianbao Xie and Zesen Cheng and Hang Zhang and Zhibo Yang and Haiyang Xu and Junyang Lin},
      year={2025},
      eprint={2502.13923},
      archivePrefix={arXiv},
      primaryClass={cs.CV},
      url={https://arxiv.org/abs/2502.13923}, 
}

@misc{monoqwen,
  author       = {{lightonai}},
  title        = {MonoQwen-Vision, the first visual document reranker},
  year         = {2024},
  url          = {https://lighton.ai/lighton-blogs/monoqwen-vision}
}

@inproceedings{visrag,
  title={Visrag: Vision-based retrieval-augmented generation on multi-modality documents},
  author={Yu, Shi and Tang, Chaoyue and Xu, Bokai and Cui, Junbo and Ran, Junhao and Yan, Yukun and Liu, Zhenghao and Wang, Shuo and Han, Xu and Liu, Zhiyuan and others},
  booktitle={International Conference on Learning Representations},
  volume={2025},
  pages={21074--21098},
  year={2025}
}

@article{m3docrag,
  title={M3docrag: Multi-modal retrieval is what you need for multi-page multi-document understanding},
  author={Cho, Jaemin and Mahata, Debanjan and Irsoy, Ozan and He, Yujie and Bansal, Mohit},
  journal={arXiv preprint arXiv:2411.04952},
  year={2024}
}

@inproceedings{vdocrag,
  title={Vdocrag: Retrieval-augmented generation over visually-rich documents},
  author={Tanaka, Ryota and Iki, Taichi and Hasegawa, Taku and Nishida, Kyosuke and Saito, Kuniko and Suzuki, Jun},
  booktitle={2025 IEEE/CVF Conference on Computer Vision and Pattern Recognition (CVPR)},
  pages={24827--24837},
  year={2025},
  organization={IEEE}
}

@inproceedings{mmdocir,
  title={MMDocIR: Benchmarking multimodal retrieval for long documents},
  author={Dong, Kuicai and Chang, Yujing and Deik, Derrick Goh Xin and Li, Dexun and Tang, Ruiming and Liu, Yong},
  booktitle={Proceedings of the 2025 Conference on Empirical Methods in Natural Language Processing},
  pages={30959--30993},
  year={2025}
}

@misc{linq,
      title={Linq-Embed-Mistral Technical Report}, 
      author={Chanyeol Choi and Junseong Kim and Seolhwa Lee and Jihoon Kwon and Sangmo Gu and Yejin Kim and Minkyung Cho and Jy-yong Sohn},
      year={2024},
      eprint={2412.03223},
      archivePrefix={arXiv},
      primaryClass={cs.CL},
      url={https://arxiv.org/abs/2412.03223}, 
}

@article{qwen3emb,
  title={Qwen3 embedding: Advancing text embedding and reranking through foundation models},
  author={Zhang, Yanzhao and Li, Mingxin and Long, Dingkun and Zhang, Xin and Lin, Huan and Yang, Baosong and Xie, Pengjun and Yang, An and Liu, Dayiheng and Lin, Junyang and others},
  journal={arXiv preprint arXiv:2506.05176},
  year={2025}
}
\bibliographystyle{plainnat}

\appendix
\section{Proofs}
\label{sec:proofs}

Throughout, a candidate set has $K$ members and $\langle x, y\rangle_K = \tfrac{1}{K}\sum_{i} x_i y_i$. A vector is standardized when $\langle x, \mathbf{1}\rangle_K = 0$ and $\langle x, x\rangle_K = 1$, which is what Equation \ref{eq:zscore} produces, and $\mathbf{y}$ is standardized the same way. All three of $z(\mathbf{b})$, $z(\mathbf{s})$ and $\mathbf{y}$ are therefore unit vectors in this inner product, so an inner product between any two of them is their correlation.

\subsection{Proposition \ref{prop:weight}}
\label{sec:proof1}

\textbf{The objective.} Write $\alpha = 1-w$ and $\beta = w$, so that $\mathbf{f} = \alpha\, z(\mathbf{b}) + \beta\, z(\mathbf{s})$. Its correlation with $\mathbf{y}$ is the inner product of $\mathbf{y}$ with $\mathbf{f}$ normalized to unit length, and $\mathbf{f}$ is centered because both of its terms are, so no mean has to be removed. Bilinearity gives the numerator

\begin{equation}
  \langle \mathbf{f}, \mathbf{y}\rangle_K
  \;=\; \alpha\, \langle z(\mathbf{b}), \mathbf{y}\rangle_K
      + \beta\, \langle z(\mathbf{s}), \mathbf{y}\rangle_K
  \;=\; \alpha\, c_b + \beta\, c_s ,
\end{equation}

and, using $\langle z(\mathbf{b}), z(\mathbf{b})\rangle_K = \langle z(\mathbf{s}), z(\mathbf{s})\rangle_K = 1$ and $\langle z(\mathbf{b}), z(\mathbf{s})\rangle_K = \rho$, the squared norm

\begin{equation}
  \langle \mathbf{f}, \mathbf{f}\rangle_K
  \;=\; \alpha^{2} + \beta^{2} + 2\alpha\beta\rho .
\end{equation}

Collecting $v = (\alpha, \beta)^{\top}$, $c = (c_b, c_s)^{\top}$ and $\Sigma = \left(\begin{smallmatrix} 1 & \rho \\ \rho & 1 \end{smallmatrix}\right)$, the correlation to be maximized is the ratio

\begin{equation}
  \mathcal{R}(v) \;=\; \frac{v^{\top} c}{\sqrt{v^{\top} \Sigma\, v}} .
  \label{eq:rayleigh}
\end{equation}

\textbf{The maximizer.} The eigenvalues of $\Sigma$ are $1 \pm \rho$, so $|\rho|<1$ makes $\Sigma$ positive definite and $\Sigma^{1/2}$ invertible. Substituting $u = \Sigma^{1/2} v$ turns Equation \ref{eq:rayleigh} into $u^{\top} \Sigma^{-1/2} c / \lVert u \rVert$, which by the Cauchy Schwarz inequality is at most $\lVert \Sigma^{-1/2} c\rVert = \sqrt{c^{\top} \Sigma^{-1} c}$, with equality exactly when $u$ is a positive multiple of $\Sigma^{-1/2} c$. Undoing the substitution, the maximizers are the positive multiples of

\begin{equation}
  \Sigma^{-1} c \;=\; \frac{1}{1-\rho^{2}}
  \begin{pmatrix} 1 & -\rho \\ -\rho & 1 \end{pmatrix}
  \begin{pmatrix} c_b \\ c_s \end{pmatrix}
  \;\propto\;
  \begin{pmatrix} c_b - \rho\, c_s \\ c_s - \rho\, c_b \end{pmatrix} .
  \label{eq:direction}
\end{equation}

\textbf{The normalization.} Equation \ref{eq:rayleigh} is invariant to the scale of $v$, so the whole ray of Equation \ref{eq:direction} attains the maximum and the constraint $\alpha + \beta = 1$ of Equation \ref{eq:fusion} selects one point of it. The two components of Equation \ref{eq:direction} sum to

\begin{equation}
  (c_b - \rho\, c_s) + (c_s - \rho\, c_b) \;=\; (c_b + c_s)(1-\rho),
\end{equation}

which is positive since both components are positive by assumption. Normalizing by this sum gives positive weights with $\alpha+\beta=1$, so the maximizing direction is feasible and its second component is

\begin{equation}
  w^{\star} \;=\; \frac{c_s - \rho\, c_b}{(c_b + c_s)(1-\rho)},
\end{equation}

which is Equation \ref{eq:wstar}.

\textbf{The endpoint.} Since the denominator is positive, $w^{\star} < 1$ is equivalent to $c_s - \rho\, c_b < (c_b + c_s)(1-\rho)$. Expanding the right hand side gives $c_b + c_s - \rho\, c_b - \rho\, c_s$, and cancelling $c_s - \rho\, c_b$ from both sides leaves $0 < c_b - \rho\, c_s$, which is the stated condition. The same computation on the first component gives $w^{\star} > 0$ if and only if $c_s > \rho\, c_b$, so the optimum is interior exactly when each score agrees with relevance by more than a $\rho$ scaled copy of the other would. \qed

\textbf{Weighted positions.} Fix $\omega \in \mathbb{R}^{K}$ with every $\omega_i > 0$ and replace $\langle \cdot, \cdot\rangle_K$ by $\langle x, y\rangle_{\omega} = \big(\sum_i \omega_i\big)^{-1} \sum_i \omega_i\, x_i y_i$, standardizing every vector and defining $c_b$, $c_s$ and $\rho$ in that inner product. Positive weights make $\langle \cdot, \cdot\rangle_{\omega}$ symmetric, bilinear and positive definite, which is everything the argument above uses. The numerator expands by bilinearity, the squared norm by bilinearity and symmetry, and the Cauchy Schwarz step holds in any inner product space. Every line therefore goes through unchanged and the statement holds verbatim. The deployed choice $\omega_i = 1/\log_2(1+\text{rank}_i)$ is the discount of the ranking metric itself, taken under the order the retriever supplies, so it is fixed before any weight is chosen and carries no free parameter.

\subsection{Proposition \ref{prop:bound}}
\label{sec:proof2}

\textbf{A positive semidefinite matrix.} Fix a candidate set and collect the three standardized vectors $z(\mathbf{s})$, $z(\mathbf{t})$ and $\mathbf{y}$. Their matrix of pairwise inner products is

\begin{equation}
  R \;=\; \begin{pmatrix} 1 & r & c_s \\ r & 1 & c_t \\ c_s & c_t & 1 \end{pmatrix},
\end{equation}

whose diagonal is one because each vector is standardized. For any $a \in \mathbb{R}^{3}$, bilinearity gives $a^{\top} R\, a = \big\lVert a_1 z(\mathbf{s}) + a_2 z(\mathbf{t}) + a_3 \mathbf{y} \big\rVert^{2}_{K} \ge 0$, so $R$ is positive semidefinite and in particular $\det R \ge 0$.

\textbf{A quadratic in $c_s$.} Expanding along the first row,

\begin{equation}
  \det R \;=\; \big(1 - c_t^{2}\big) - r\big(r - c_t c_s\big) + c_s\big(r c_t - c_s\big)
          \;=\; 1 - r^{2} - c_t^{2} - c_s^{2} + 2\, r\, c_t\, c_s .
\end{equation}

Requiring this to be non-negative and changing sign gives

\begin{equation}
  c_s^{2} \;-\; 2\, r\, c_t\, c_s \;+\; \big(r^{2} + c_t^{2} - 1\big) \;\le\; 0 .
\end{equation}

The left hand side is a quadratic in $c_s$ with positive leading coefficient, so the inequality confines $c_s$ to the closed interval between its two roots. Those roots are

\begin{equation}
  r\, c_t \;\pm\; \sqrt{r^{2} c_t^{2} - r^{2} - c_t^{2} + 1}
  \;=\; r\, c_t \;\pm\; \sqrt{\big(1 - r^{2}\big)\big(1 - c_t^{2}\big)},
\end{equation}

the discriminant factoring because $r^{2} c_t^{2} - r^{2} - c_t^{2} + 1 = (1-r^{2})(1-c_t^{2})$. The lower root is Equation \ref{eq:bound}, and the upper root is the matching upper bound.

\textbf{Monotonicity.} Write $g(r) = r\, c_t - \sqrt{(1-r^{2})(1-c_t^{2})}$ for the lower root at fixed $c_t$. Differentiating the square root through the chain rule,

\begin{equation}
  g'(r) \;=\; c_t \;+\; \frac{r\big(1 - c_t^{2}\big)}{\sqrt{\big(1-r^{2}\big)\big(1-c_t^{2}\big)}}
        \;=\; c_t \;+\; r\, \sqrt{\frac{1 - c_t^{2}}{1 - r^{2}}} ,
\end{equation}

whose two terms are both positive for $r, c_t \in (0,1)$. So $g$ is strictly increasing there, and raising the agreement of the model score with the target raises the floor under its agreement with relevance. \qed

\section{Additional Results}
\label{sec:extra}

\begin{table}[t]
\centering
\small
\setlength{\tabcolsep}{3.5pt}
\caption{Accuracy against measured cost. Recovery is the fraction of the headroom between the retriever and the full cross encoder that a method closes. Speedup is against the cross encoder on the same pool.}
\label{tab:efficiency_vidore2}
\begin{tabular}{llrrrr}
\toprule
Retriever & Method & NDCG@5 & Recovery & Latency (ms) & Speedup \\
\midrule
\multirow{9}{*}{ColNomic-7B} & Retriever only & 60.9 & -- & 60 & 113$\times$ \\
 & GQR & 63.8 & 24.9\% & 130 & 52$\times$ \\
 & RRF & 62.3 & 11.9\% & 93 & 72$\times$ \\
 & Score Aggregation & 62.2 & 11.4\% & 93 & 72$\times$ \\
 & Qwen3-Reranker-0.6B & 61.6 & 6.4\% & 344 & 20$\times$ \\
 & MonoQwen2-VL & 65.4 & 39.0\% & 2,860 & 2$\times$ \\
 & Ours (Qwen2.5-VL-7B) & 69.0 & 70.4\% & 156 & 44$\times$ \\
 & Ours (GLM-4.1V-9B) & 71.2 & 89.4\% & 141 & 48$\times$ \\
 & Full cross encoder & 72.4 & 100.0\% & 6,765 & 1$\times$ \\
\midrule
\multirow{8}{*}{ColPali-v1.3} & Retriever only & 50.2 & -- & 40 & 166$\times$ \\
 & RRF & 55.8 & 31.8\% & 74 & 90$\times$ \\
 & Score Aggregation & 56.8 & 37.9\% & 74 & 90$\times$ \\
 & Qwen3-Reranker-0.6B & 54.7 & 25.7\% & 330 & 20$\times$ \\
 & MonoQwen2-VL & 62.6 & 71.4\% & 2,808 & 2$\times$ \\
 & Ours (Qwen2.5-VL-7B) & 64.9 & 84.7\% & 136 & 49$\times$ \\
 & Ours (GLM-4.1V-9B) & 66.1 & 91.3\% & 122 & 55$\times$ \\
 & Full cross encoder & 67.6 & 100.0\% & 6,694 & 1$\times$ \\
\bottomrule
\end{tabular}
\end{table}

\begin{table}[t]
\centering
\small
\setlength{\tabcolsep}{4pt}
\caption{The same arms under four further metrics, macro over the corpora of each benchmark. The fusion weight of every arm is set as in Tables \ref{tab:ndcg5_vidore2} and \ref{tab:ndcg5_vidore3} and then held fixed, so none of these four metrics had a part in choosing it.}
\label{tab:metrics}
\resizebox{\textwidth}{!}{%
\begin{tabular}{lllrrrrr}
\toprule
Benchmark & Retriever & Method & NDCG@5 & NDCG@10 & R@5 & R@10 & MRR@10 \\
\midrule
\multirow{17}{*}{ViDoRe 2} & \multirow{9}{*}{ColNomic-7B} & Retriever only & 60.9 & 63.8 & 58.0 & 69.8 & 71.3 \\
 & & GQR & 63.8 & 65.2 & 58.6 & 68.5 & 73.8 \\
 & & RRF & 62.3 & 64.8 & 57.6 & 69.9 & 72.5 \\
 & & Score Aggregation & 62.2 & 65.1 & 57.5 & 69.4 & 72.5 \\
 & & Qwen3-Reranker-0.6B & 61.6 & 64.2 & 58.0 & 69.0 & 71.5 \\
 & & MonoQwen2-VL & 65.4 & 67.9 & 60.3 & 72.3 & 75.7 \\
 & & Ours (Qwen2.5-VL-7B) & 69.0 & 70.5 & 63.9 & 73.8 & 78.0 \\
 & & Ours (GLM-4.1V-9B) & 71.2 & 72.2 & 64.4 & 73.6 & 80.4 \\
 & & Full cross encoder & 72.4 & 73.3 & 65.3 & 73.8 & 82.2 \\
\cmidrule(lr){2-8}
 & \multirow{8}{*}{ColPali-v1.3} & Retriever only & 50.2 & 53.3 & 47.8 & 60.2 & 61.5 \\
 & & RRF & 55.8 & 57.6 & 52.9 & 63.5 & 66.2 \\
 & & Score Aggregation & 56.8 & 58.7 & 52.8 & 63.4 & 67.3 \\
 & & Qwen3-Reranker-0.6B & 54.7 & 56.9 & 52.8 & 64.0 & 64.6 \\
 & & MonoQwen2-VL & 62.6 & 63.7 & 57.1 & 66.4 & 73.3 \\
 & & Ours (Qwen2.5-VL-7B) & 64.9 & 65.0 & 59.3 & 67.0 & 75.0 \\
 & & Ours (GLM-4.1V-9B) & 66.1 & 66.3 & 58.9 & 66.7 & 76.5 \\
 & & Full cross encoder & 67.6 & 67.6 & 60.9 & 67.9 & 77.8 \\
\midrule
\multirow{17}{*}{ViDoRe 3} & \multirow{9}{*}{ColNomic-7B} & Retriever only & 52.8 & 55.3 & 48.2 & 59.8 & 67.0 \\
 & & GQR & 54.3 & 56.8 & 49.4 & 61.3 & 68.6 \\
 & & RRF & 51.4 & 54.6 & 47.3 & 60.4 & 66.2 \\
 & & Score Aggregation & 53.5 & 56.0 & 48.8 & 60.6 & 67.6 \\
 & & Qwen3-Reranker-0.6B & 54.0 & 56.7 & 49.4 & 61.3 & 68.3 \\
 & & MonoQwen2-VL & 54.5 & 56.8 & 49.7 & 60.8 & 68.5 \\
 & & Ours (Qwen2.5-VL-7B) & 56.2 & 58.2 & 51.6 & 62.3 & 69.4 \\
 & & Ours (GLM-4.1V-9B) & 55.8 & 57.8 & 50.5 & 61.5 & 69.8 \\
 & & Full cross encoder & 58.0 & 59.8 & 52.6 & 63.2 & 71.5 \\
\cmidrule(lr){2-8}
 & \multirow{8}{*}{ColPali-v1.3} & Retriever only & 38.8 & 41.3 & 36.4 & 46.5 & 51.8 \\
 & & RRF & 42.9 & 45.6 & 39.5 & 50.6 & 57.3 \\
 & & Score Aggregation & 44.6 & 46.4 & 40.8 & 50.2 & 59.4 \\
 & & Qwen3-Reranker-0.6B & 47.8 & 49.0 & 43.3 & 51.7 & 62.0 \\
 & & MonoQwen2-VL & 46.7 & 48.3 & 41.9 & 51.0 & 61.4 \\
 & & Ours (Qwen2.5-VL-7B) & 48.8 & 50.0 & 44.2 & 52.4 & 63.1 \\
 & & Ours (GLM-4.1V-9B) & 48.8 & 49.8 & 43.8 & 51.9 & 63.2 \\
 & & Full cross encoder & 51.1 & 51.7 & 45.5 & 53.0 & 65.6 \\
\bottomrule
\end{tabular}}
\end{table}

Ranking quality under four further metrics and the measured cost of each arm are reported here.

\subsection{Fusion at Full Depth}

Proposition \ref{prop:weight} concerns any pair of scores, so it applies to the uncompressed cross encoder as much as to the readout. \Cref{tab:fullce_w} sweeps the fusion weight over the full cross encoder on all 24 subsets. No subset attains its optimum at $w = 1$, the median optimum is 0.80, and ranking by the cross encoder alone loses 1.5 points on average and up to 4.3. Keeping the retriever score in the order is therefore no correction for compression, since it improves the very model the readout reproduces.

\begin{table}[h]
\centering
\small
\caption{The full cross encoder fused with the retriever, the optimal weight located per subset on a grid of step 0.05.}
\label{tab:fullce_w}
\resizebox{\textwidth}{!}{%
\begin{tabular}{llrrrrr}
\toprule
Benchmark & Retriever & Subsets & Optimal at $w = 1$ & Median optimal $w$ & NDCG@5 at $w = 1$ & NDCG@5 at optimum \\
\midrule
ViDoRe 2 & ColNomic-7B & 4 & 0 & 0.850 & 72.0 & 73.6 \\
ViDoRe 2 & ColPali-v1.3 & 4 & 0 & 0.875 & 67.0 & 67.9 \\
ViDoRe 3 & ColNomic-7B & 8 & 0 & 0.700 & 55.9 & 58.4 \\
ViDoRe 3 & ColPali-v1.3 & 8 & 0 & 0.825 & 50.7 & 51.5 \\
\bottomrule
\end{tabular}}
\end{table}

\subsection{Number of Calibration Queries}

\Cref{fig:calib} draws the calibration queries at random from the corpora a fold is fitted on, five draws per size, and scores the held out corpus of the 154 query calibration sample. None of these queries carries a relevance judgement. Sixteen queries already reach 74.3 on GLM and 71.1 on Qwen, 99.4 and 98.3 percent of the full sample, and both backbones stay above the 67.2 that GQR scores on the same queries at every size.

\begin{figure}[h]
\centering
\includegraphics[width=0.5\textwidth]{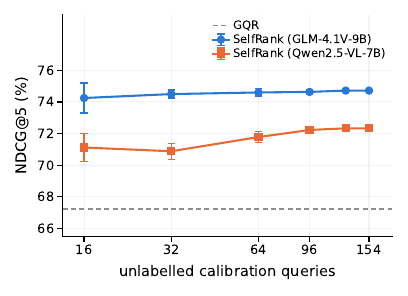}
\caption{NDCG@5 against the number of unlabelled calibration queries, mean and standard deviation over five draws.}
\label{fig:calib}
\end{figure}

\subsection{Label Free Regularization}
\label{sec:lambda}

\Cref{tab:lambda} compares relevance judgements with centered teacher error for selecting $\lambda$. Within each outer fold, two corpora fit the readout and the third validates it; the modal choice across four folds is deployed. On GLM both rules deploy $10^{3}$, yielding bit identical refitted vectors despite different fold choices. On Qwen the teacher rule selects $10^{3}$ in every fold, scoring 72.3 against 71.3 for the labelled rule. Reported results use the teacher rule.

\begin{table}[h]
\centering
\small
\caption{The ridge strength each rule selects on the four folds, the modal value deployed, and NDCG@5 on the calibration sample at that value.}
\label{tab:lambda}
\resizebox{\textwidth}{!}{%
\begin{tabular}{llccccc r}
\toprule
Backbone & Rule & Biomed & Econ & ESG-H & ESG-F & Deployed & NDCG@5 \\
\midrule
GLM-4.1V-9B & Relevance judgements & $10^{3}$ & $10^{4}$ & $10^{3}$ & $10^{1}$ & $10^{3}$ & 74.7 \\
 & Teacher alone & $10^{3}$ & $10^{4}$ & $10^{3}$ & $10^{3}$ & $10^{3}$ & 74.7 \\
\midrule
Qwen2.5-VL-7B & Relevance judgements & $10^{3}$ & $10^{4}$ & $10^{4}$ & $10^{1}$ & $10^{4}$ & 71.3 \\
 & Teacher alone & $10^{3}$ & $10^{3}$ & $10^{3}$ & $10^{3}$ & $10^{3}$ & 72.3 \\
\bottomrule
\end{tabular}}
\end{table}

\subsection{Pool Depth for Every Method}

\Cref{tab:depth_all} extends \Cref{fig:pooldepth} to every reranker of Table \ref{tab:ndcg5_vidore2}, \method{} setting its weight by Equation \ref{eq:wstar} and each baseline by NDCG@5 on the same inner corpus, at every depth. At every depth from five to a hundred candidates both \method{} backbones rank above all three external baselines, by at least 5.0 points from ten candidates onward.

\begin{table}[h]
\centering
\small
\caption{NDCG@5 on ViDoRe 2 with ColNomic candidates against the number of candidates reranked. The retriever alone scores 60.9 and GQR 63.8 at every depth.}
\label{tab:depth_all}
\resizebox{\textwidth}{!}{%
\begin{tabular}{rrrrrr}
\toprule
$K$ & RRF & Score Aggregation & Qwen3-Reranker-0.6B & \method{} (Qwen2.5-VL-7B) & \method{} (GLM-4.1V-9B) \\
\midrule
5 & 61.9 & 61.4 & 60.8 & 63.7 & 65.2 \\
10 & 62.8 & 62.2 & 60.3 & 67.8 & 69.5 \\
15 & 62.6 & 62.5 & 60.5 & 68.9 & 70.8 \\
20 & 62.3 & 62.2 & 61.6 & 69.0 & 71.1 \\
25 & 62.0 & 62.5 & 60.8 & 69.2 & 71.1 \\
30 & 61.9 & 62.3 & 61.1 & 69.4 & 71.0 \\
40 & 61.5 & 62.6 & 61.1 & 69.5 & 71.0 \\
50 & 61.5 & 62.6 & 61.0 & 69.6 & 70.8 \\
60 & 61.6 & 62.5 & 61.7 & 69.6 & 70.4 \\
70 & 61.7 & 62.8 & 61.7 & 69.6 & 70.8 \\
80 & 61.5 & 62.9 & 61.9 & 69.8 & 70.6 \\
90 & 61.5 & 64.1 & 62.6 & 70.1 & 70.5 \\
100 & 61.5 & 64.0 & 62.1 & 70.2 & 70.5 \\
\bottomrule
\end{tabular}}
\end{table}

\subsection{Operating Point Without Relevance Judgements}
\label{sec:labelfree_w}

Proposition \ref{prop:weight} requires the correlation of each score with relevance, and \Cref{tab:labelfree_w} shows that both can be recovered from the scores alone. Under a one factor model in which the retriever score, the readout and the full depth margin of the uncompressed model each carry relevance plus mutually uncorrelated noise, $c_b=\sqrt{\rho_{bs}\rho_{bt}/\rho_{st}}$ and $c_s=\sqrt{\rho_{bs}\rho_{st}/\rho_{bt}}$, where each $\rho$ is a correlation between two scores and involves no relevance judgement. Read on the unlabelled queries of the corpus being ranked, this estimate reaches 59.9 on average against 60.1 for the deployed rule, and sixteen unlabelled queries per corpus already give 59.8.

\begin{table}[h]
\centering
\small
\caption{NDCG@5 (\%) of \method{} under three ways of setting the fusion weight. Sweep picks $w$ by NDCG@5 on the next corpus of the same benchmark. Proposition \ref{prop:weight} reads the three correlations on that corpus with its relevance judgements, as in \Cref{tab:ndcg5_vidore2,tab:ndcg5_vidore3}. Label free reads only correlations between scores, on the unlabelled queries of the corpus being ranked, using all of them or sixteen drawn at random with mean and standard deviation over five draws. The third score is the full depth margin of Qwen2.5-VL-7B.}
\label{tab:labelfree_w}
\begin{tabular*}{\textwidth}{@{\extracolsep{\fill}}lllrrrr}
\toprule
 & & & & & \multicolumn{2}{c}{Label free} \\
\cmidrule(lr){6-7}
Benchmark & Retriever & Backbone & Sweep & Proposition \ref{prop:weight} & All queries & 16 queries \\
\midrule
\multirow{4}{*}{ViDoRe 2} & \multirow{2}{*}{ColNomic-7B} & Qwen2.5-VL-7B & 69.6 & 69.0 & 69.4 & 69.4 {\scriptsize $\pm$0.2} \\
 & & GLM-4.1V-9B & 71.3 & 71.2 & 70.7 & 70.2 {\scriptsize $\pm$0.4} \\
 & \multirow{2}{*}{ColPali-v1.3} & Qwen2.5-VL-7B & 65.1 & 64.9 & 65.2 & 65.2 {\scriptsize $\pm$0.1} \\
 & & GLM-4.1V-9B & 65.8 & 66.1 & 66.3 & 65.8 {\scriptsize $\pm$0.3} \\
\midrule
\multirow{4}{*}{ViDoRe 3} & \multirow{2}{*}{ColNomic-7B} & Qwen2.5-VL-7B & 56.2 & 56.2 & 54.7 & 54.6 {\scriptsize $\pm$0.1} \\
 & & GLM-4.1V-9B & 55.5 & 55.8 & 55.9 & 55.8 {\scriptsize $\pm$0.1} \\
 & \multirow{2}{*}{ColPali-v1.3} & Qwen2.5-VL-7B & 48.6 & 48.8 & 48.5 & 48.5 {\scriptsize $\pm$0.1} \\
 & & GLM-4.1V-9B & 48.6 & 48.8 & 48.8 & 48.7 {\scriptsize $\pm$0.0} \\
\midrule
\multicolumn{3}{l}{Average} & 60.1 & 60.1 & 59.9 & 59.8 \\
\bottomrule
\end{tabular*}
\end{table}

\end{document}